\documentclass[10pt]{article}

\usepackage{authblk}
\usepackage{amsmath, amssymb, amsfonts, amsthm}
\usepackage[british,UKenglish,USenglish,english,american]{babel}
\usepackage{amsmath}
\usepackage{amssymb}
\usepackage{url}
\usepackage{amsthm}
\usepackage{amsxtra}
\usepackage{mathrsfs}
\usepackage{fancyhdr}
\usepackage{enumerate}
\usepackage{graphicx}
\usepackage{tikz}

\usepackage{amsthm,hyperref
}
\numberwithin{equation}{section} 
\newtheorem{theorem}{Theorem}[section] 
 
\newtheorem{proposition}[theorem]{Proposition} 
\newtheorem{corollary}[theorem]{Corollary} 
\theoremstyle{definition}

\newtheorem{definition}[theorem]{Definition}

\newtheorem{assumption}[theorem]{Assumption}

\newcommand{\beq}{\begin{equation}}
\newcommand{\eeq}{\end{equation}}
\newcommand{\be}{\begin{equation*}}
\newcommand{\ee}{\end{equation*}}
\newcommand{\n}{\noindent}

\newcommand{\RE}{\mathbb R}
\newcommand{\erre}{\mathbb R}

\newcommand{\BB}{\mathscr B}
\newcommand{\B}{\mathcal B}
\newcommand{\DD}{\mathscr D}

\newcommand{\MM}{\mathcal{M}}

\newcommand{\NN}{\mathcal{N}}

\newcommand{\lf}{\left}
\newcommand{\ri}{\right}

\newcommand{\al}{\alpha}

\newcommand{\la}{\lambda}
\newcommand{\de}{\delta}

\newcommand{\xx}{\langle x \rangle}

\newcommand{\lan}{\langle}
\newcommand{\ran}{\rangle}
\renewcommand{\Im}{\operatorname{Im}\,}

\renewcommand{\leq}{\leqslant}
\renewcommand{\geq}{\geqslant}

\newcommand{\EE}{\mathcal E}
\newcommand{\HH}{\mathcal H}

\newcommand{\WW}{\mathcal W}

\newcommand{\xv}{\mathbf{x}}
\newcommand{\kv}{\mathbf{k}}

\newcommand{\kvp}{\mathbf{k}^{\prime}}

\newcommand{\Qb}{\overline Q}

\newcommand{\bdm}{\begin{displaymath}}
\newcommand{\edm}{\end{displaymath}}
\newcommand{\bdn}{\begin{eqnarray}}
\newcommand{\edn}{\end{eqnarray}}
\newcommand{\bay}{\begin{array}{c}}
\newcommand{\eay}{\end{array}}
\newcommand{\ben}{\begin{enumerate}}
\newcommand{\een}{\end{enumerate}}
\newcommand{\beqn}{\begin{eqnarray}}
\newcommand{\eeqn}{\end{eqnarray}}

\newcommand{\R}{\mathbb{R}}
\newcommand{\N}{\mathbb{N}}

\newcommand{\diff}{\mathrm{d}}

\newcommand{\amp}{A_V(\kv , \kvp)}
\newcommand{\ampe}{A_{\mathrm{eff}}(\kv , \kvp; \la)}
\newcommand{\eigenv}{\phi_{V,\kv}}
\newcommand{\aeff}{a_{\mathrm{eff}}(\lambda)}

\newcommand{\tx}{\textstyle}

\newcommand{\bra}[1]{\lf\langle #1\ri|}
\newcommand{\ket}[1]{\lf|#1 \ri\rangle}
\newcommand{\braket}[2]{\lf\langle #1|#2 \ri\rangle}

\newcommand{\mean}[3]{\bra{#1}#2\ket{#3}}

\begin{document}
\title{Expansion of the resolvent in a Feshbach model}

\author[1]{Raffaele Carlone}
\affil[1]{Dipartimento di Matematica e Applicazioni ``R. Caccioppoli'', Universit\'a di Napoli ``Federico II'', Via Cinthia, Monte S. Angelo, 80126 Napoli, Italy}

\author[2]{Domenico Finco}
\affil[2]{Facolt\`a di Ingegneria, Universit\`a Telematica Internazionale Uninettuno, Corso V. Emanuele II 39,  00186 Roma, Italy}



\maketitle

\begin{abstract}
\noindent
In this paper we extend the results proved in (\cite{ccft}) about Feshbach resonances in a multichannel Hamiltonian $\mathcal{H}$, proving a low energy expansion of the resolvent $(\mathcal{H}-k^{2})^{-1}$  as $k\to 0$ in the resonant case. 
\end{abstract}	

\begin{center}
\emph{Dedicated to Gianfausto Dell'Antonio on the occasion of his 85th birthday}
\end{center}

\section{Introduction}

\noindent
The physics of ultracold quantum gases and molecular quantum gases is a research which had had a steep growth in the last twenty years induced by incredible progresses, 
both from the experimental point of view and from the theoretical one.
The wide range of applications covers atomic and molecular physics, condensed matter and few and many-body physics.  

A major breakthrough in this field was the experimental observation of a dilute Bose gas condensation in 1995 (\cite {AEMWC}) as in the Einstein predictions of 1925 .
After this fundamental experimental realization, a big effort was made for a deep understanding of the physical processes involved.
A crucial point was to go beyond the mean field physics, to study interactions between ultra cold atoms and observe long and short range correlation phenomena (\cite{KZ}).   
The extraordinary degree of control needed on such systems in order to reach these extreme conditions,  was really challenging.

From the experimental point of view, besides all the different techniques to control the physical properties of condensate, two of them in particular had  greater success: one is the realization of optical lattices of different spatial dimensions (\cite{B},\cite{Ch}), while the second is the tunability of the scattering length using  Feshbach resonances (see \cite{CGJT} and \cite{IASMSK} for a review  ).

In a dilute quantum gas  the density and mobility condition are such that only the two-body scattering processes are relevant. Moreover in many situations, since the typical energies allow 
only elastic scattering, the most relevant parameter is the scattering length and the possibility to tune it provides an effective mean to control the interaction.

A key feature of Feshbach resonances is the presence of an open channel (where scattering processes are allowed) and a closed channel (where scattering processes are forbidden) with bound states.
It may happen that a bound state of the closed channel crosses the ionization threshold of the open channel, that is the bottom of the continuous spectrum of the hamiltonian
of the open channel, and strongly interacts deeply influencing the scattering process, see \cite{fis1}, even when the interaction between channel is weak. 
This is often understood in the physical literature as virtual scattering
process with a metastable state at the bottom of the spectrum.

\noindent

 \begin{figure}\label{fig}
\begin{center}
 \includegraphics[height=7cm]{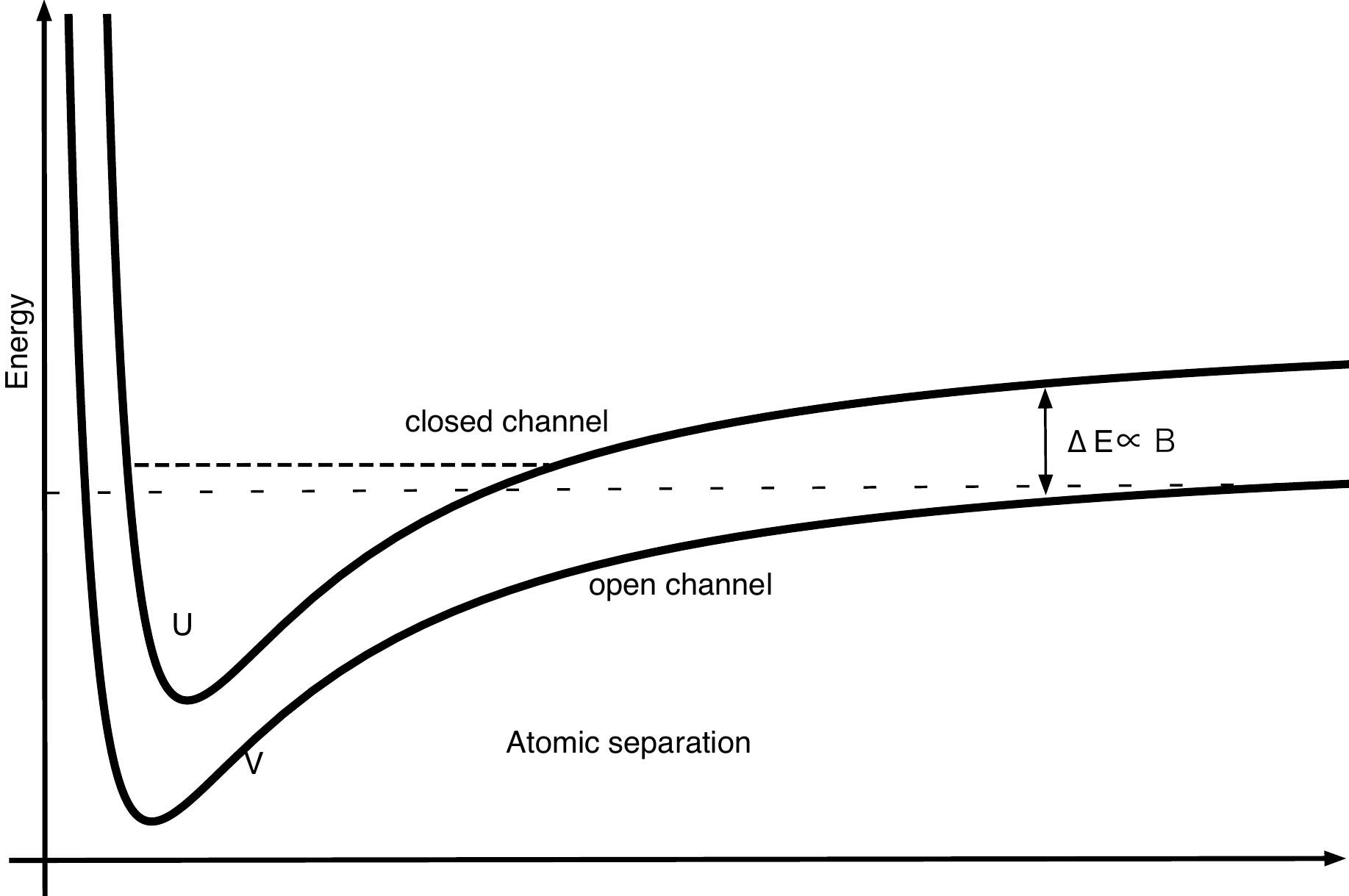}
 \caption{\footnotesize{In the figure are plotted the potentials for an open channel and a closed channel with a bound state as function of atomic separation for alkali atoms. } }
\end{center}
 \end{figure}

\noindent
In modern realizations the splitting is realized by applying an external magnetic field $B$ to, for instance, alkali atoms and it is due to the coupling of states with different spin quantum numbers
with the magnetic field. In this way the energy splitting $\Delta E$ and energy levels of the bound states of the closed channel can be externally controlled and it is possible to realize
the conditions for a Feshbach resonance.

\noindent In the physics literature it was obtained a formula relating the scattering length with the magnetic field around the value of the Feshbach resonance (\cite{PS})
$$ a(B)=a \left(1-\frac{\Delta}{B-B_{0}}\right)$$
where $a$ is the scattering length far from the resonant values, $B_{0}$ is the value of the magnetic field at the resonance and $\Delta$ is  the resonance width.

\noindent 
The rigorous analysis of resonances  has a very long story (\cite{E}):  the first mathematical model for multichannel hamiltonian traces back to the original paper of Friedrichs (\cite{fiss1}) in 1948. Many mathematical features of resonances were studied using different strategies, 
as the dilation-analytic technique (\cite{HS}) or simplified models using point interactions (\cite{ccf},\cite{ccf2},\cite{ccf3}).  

Recently in \cite{ccft} the formation of Feshbach resonances and the dependence of the scattering length from the magnetic field were rigorously studied, proving results about the existence and localization
of resonances and the behavior of the scattering length near a resonance.

In section 2, we fix the notation and we recall the main results of \cite{ccft}. In section 3 we prove a low energy expansion of the resolvent of a matrix hamiltonian when
a Feshbach resonance is present; our main result is Theorem \ref{ourmain}.

\section{Existence of Feshbach Resonances}

In this section we fix some notational conventions used in the paper and we recall the main results of cite{ccft}.

\n
Bold face letters denote vectors in $ \R^3 $, e.g., $ \xv $, while scalars are denoted by regular letters, e.g., $ E $. When there is no ambiguity we also use the notation $ x : = |\xv| $ for the modulus of a vector $ \xv $.

\n
Since we always deal with functions on $ \R^3 $, we often omit the base space in Banach space notation, i.e., $ L^p : = L^p(\R^3) $; moreover when there is no possible confusion we denote $ \lf\| \: \cdot \: \ri\|_{L^p} =: \lf\| \: \cdot \: \ri\|_p $.

\noindent
We recall the definition of weighted Hilbert spaces: set $ \lan x \ran : = \sqrt{1 + x^2} $ for short, then, for any $ s \geq 0 $, we define 
\beq
	\label{eq:wnorm}
	\lf\| f \ri\|_{L^2_s} : = \lf\| \lan x \ran^s f \ri\|_2.
\eeq

\n
The closure of $ C^{\infty}_0 $ w.r.t. the above norm is denoted by $ L^2_s $. The weighted Sobolev space $ H^2_s $, $ s \geq 0 $, is defined analogously as the closure of $ C^{\infty}_0 $ w.r.t. the norm
\beq
	\label{eq:wSnorm}
	\lf\| f \ri\|_{H^2_s} : = \lf\| f \ri\|_{L^2_s} + \lf\| \Delta f \ri\|_{L^2_s}.
\eeq
The conventional Sobolev spaces $ H^p $, $ p \in \R $ can be defined via Fourier transform as the closure of $ C^{\infty}_0(\R^3) $ w.r.t. the norms 
\beq
	\lf\| f \ri\|_{H^p} : = \big\| \lan k \ran^p \hat{f} \big\|_2,
\eeq
 where we use the following convention for the Fourier transform
\beq
	\hat{f}(\kv) : = \frac{1}{(2\pi)^{3/2}} \int_{\R^3} \diff \xv \: e^{-i \kv \cdot \xv} f(\xv).
\eeq
By the properties of the Fourier transform and a simple exchange of the role of $ \xv $ and $ \kv $, one easily gets
\beq\label{LsHs}
	 \lf\| f \ri\|_{L^2_s}^2 = \int_{\R^3} \diff \xv \: \lf(1 +  x^2 \ri)^{s} \lf| f(\xv) \ri|^2 = \int_{\R^3} \diff \xv \: \lf( 1 + x^2 \ri)^{s} \Big| \widehat{\hat{f}}(-\xv) \Big|^2 = \big\| \hat{f} \big\|_{H^s}^2.
\eeq
Hence, we obtain the useful identity
\beq
	\label{eq:wSnorm-identity}
	\lf\| f \ri\|_{H^2_s} = \big\| \hat{f} \big\|_{H^s} + \big\| k^2 \hat{f} \big\|_{H^s}.
\eeq

We recall some classical results on spectral and scattering theory mostly taken from \cite{A,I}, which will be used in the proofs.  We denote by $ \mathscr{B} $ the Banach space of continuous functions 
vanishing at infinity equipped with the $\sup$ norm. We also denote by $ {\mathcal B}(L^2) $ the space of bounded linear operators on $ L^2 $ and, more in general, $ \mathcal{B}(X,Y) $ stands for the space of continuous linear transformations between two Banach spaces $ X $ and $ Y $. Similarly, $ \B_0 (X,Y) $ is the space of compact operators from $ X $ to $ Y $ and $  \B_0(X) : =  \B_0(X,X) $.

Following (\cite{ccft}) we now make the mathematical setting more precise. We consider a multi-channel scattering of a particle in dimension three: we assume that there are an {\it open channel}, where the scattering is energetically possible, and a {\it closed} one where any scattering process is forbidden because of an energy constraint (see figure \ref{fig}). 

We describe the system by the following matrix Hamiltonian 
acting on the Hilbert space  ${\mathscr H} = L^2 (\RE^3) \oplus  L^2 (\RE^3) $:
\beq
\HH = 
\left(
\begin{array}{cc}
-\Delta + V & W \\
 W & -\Delta + U +\la
\end{array}
\right) = \HH_0 + \WW, 
\eeq
\beq
\HH_0 = 
\left(
\begin{array}{cc}
-\Delta + V & 0 \\
 0 & -\Delta + U +\la
\end{array}
\right),
\qquad
\WW = 
\left(
\begin{array}{cc}
0 & W \\
 W & 0
\end{array}
\right), 
\eeq
where 
\beq
	\label{eq:1bh}
	H_V=-\Delta +V,	\qquad		 H_U=-\Delta +U
\eeq

The starting point of our investigation is the eigenvalue equation for the matrix Hamiltonian $ \HH $: 
\beq\label{eigenequa}
\HH \Psi_\kv = k^2 \, \Psi_\kv
\eeq

\n
Writing $\Psi_\kv=(\varphi_\kv, \xi_\kv)^t$, where $t$ denotes transposition,  equation \eqref{eigenequa} is equivalent to the system
\beq \label{geneigen}
\begin{cases}
(-\Delta + V)\varphi_\kv + W \xi_\kv =  k^2 \varphi_\kv, \\    
( -\Delta + U +\la )\xi_\kv+ W \varphi_\kv = k^2 \xi_\kv
\end{cases} .
\eeq
Since we are interested in the low energy behavior, we can restrict to the energies  
\beq
	\label{eq:condition-k}
	0<k^2<\la.
\eeq
For $k^2 - \lambda$ in the resolvent set of $H_U$, $ -\Delta + U +\la -k^2$ has a bounded inverse $R_U(k^2-\la)$  in $\B(L^2)$ and the system \eqref{geneigen} is equivalent to the coupled integral equations
\beq \label{geneigen2}
\begin{cases}
 \varphi_\kv + R_V (k^2) \, W\, \xi_\kv = \eigenv, \\
 \xi_\kv + R_U(k^2-\la)\, W\, \varphi_\kv= 0,
\end{cases}
\eeq
where we have denoted by $ \eigenv $ the generalized eigenfunctions of $H_V$.
 The resolvent $ R_V(k^2) $ is defined as a boundary value from the upper half plane.
From \eqref{geneigen2} one sees that the problem is reduced to find the solution  $\varphi_\kv$ of the equation
\beq \label{geneigen3}
 \varphi_\kv - R_V (k^2) \, W\, R_U(k^2-\la)\, W\, \varphi_\kv = \eigenv .
\eeq

\begin{definition}[Ikebe class $ I_n $]
	\label{def:ikebe}
	\mbox{}	\\
We say that a measurable function $V$ belongs to the Ikebe class $I_n(\R^3)$, $ n \in \N $, if $V\in L^2(\R^3)$, $ V $ is locally H\"older continuous except for a finite number of points and there exists $R_0>0$ and $\de>0$ such that 
	\beq
		\label{eq:ikebe-decay}
		|V(\xv)| \leq \dfrac{c}{x^{n+\de}},	\qquad		\mbox{for } x\geq R_0.
	\eeq
\end{definition}

\begin{assumption} \label{pot}
We assume that 
	\begin{itemize}
		\item[a)] $U\in I_2(\R^3)$;
		\item[b)] $V \in I_4(\R^3) \cap L^3(\R^3)$;
		\item[c)] $W \in I_3(\R^3) \cap L^3(\R^3)$. 
	\end{itemize}
\end{assumption}

\begin{assumption} \label{boundst}
Under this assumptions $H_V$ and $H_U$ are self-adjoint operators on $H^2$.
We assume that 
\begin{itemize}
	\item[a)] $H_U  $ has $N \geq 1$ negative simple eigenvalues $ E_0 < E_1< \cdots < E_{ N-1} < 0 $, with corresponding eigenvectors $ \eta_0, \eta_1, \ldots, \eta_{N-1} \in L^2(\R^3) $;
	\item[b)]  $H_V \geq 0 $ and zero is  neither an eigenvalue nor a  resonance;
	\item[c)]  $ \ker \lf( R_V^{1/2}(0) W \ri) = \lf\{ 0 \ri\} $ in $L^2(\erre^3)$.
\end{itemize}
\end{assumption}

\noindent
We denote by $ \EE$ the set of eigenvectors related to positive eigenvalues of $ \HH $, i.e.,
\beq
	\EE : = \sigma_{\mathrm{pp}}(\HH) \cap \R^+.
\eeq
\noindent 
	\begin{proposition}[{Generalized eigenfunctions}]
		\label{eq:solution-scattering}
		\mbox{}	\\
		Let Assumption \ref{pot} hold true and let $ \lambda > 0 $ be fixed. Then, for any $ \kv \in \R^3 $ with $ k^2 \in (0,\la) \setminus \EE$ and
$k^2-\la \neq E_j$, for $j=0, \ldots, N-1$, equation \eqref{geneigen3} admits a unique continuous solution $ \varphi_\kv $, such that
$ \varphi_\kv - \eigenv \in \BB$. Furthermore, $ \varphi_\kv $ satisfies the asymptotics
		\beq
			\label{eq:varphi-asymptotics}
			\varphi_{\kv}(\xv) \underset{x \to +\infty}{\simeq} e^{i \kv \cdot \xv} + \ampe \frac{e^{i k x}}{x},
		\eeq
		with
		\beq \label{scatamp}
			\ampe = \tx\frac{1}{4 \pi} \mean{\phi_{V,\kvp}}{W\, R_U(k^2-\la)\, W}{\varphi_{\kv}} + \amp,
		\eeq
		where $ \amp $ is the scattering amplitude associated to the potential $ V $.
	\end{proposition}

	\begin{definition}[Effective scattering length]
		\label{def:scattering-length}
		\mbox{}	\\
		We define the effective scattering length in the open channel as
		\beq
			a_{\mathrm{eff}}(\la) : = \lim_{k \to 0} \ampe,
		\eeq
		where $ \ampe $ is given by \eqref{scatamp}
	\end{definition}
See \cite{ccft} for a motivation of this definition.

	\begin{theorem}[Feshbach resonances]
		\label{main}
		\mbox{}	\\
		Let Assumptions \ref{pot} and \ref{boundst} hold true and let $ \lambda > 0  $ be fixed. Then, there are at least $N$ critical values $\la_j$, $j=0,\ldots,N-1$, with $|E_j| < \la_j$, such that $ \aeff $ is continuous for $ \la \neq \la_j$ and 
		\beq
			\aeff = \frac{c_j}{\la-\la_j} + {\mathcal O}(1),  
		\eeq
		as $\la \to \la_j$, where $c_j \in \erre $.  
		\newline
		Furthermore, there is $\de_0>0$ such that, if $ \lf\| W \ri\|_3 \leq \de_0$,
		then the critical values satisfy $ \la_0 > |E_0| > \la_1 > |E_1 | > \cdots > \la_{N-1} > |E_{N-1}| $
		and any further critical value $\la_j $, with $j\geq N$, is such that $ |E_{N-1}| > \la_j >0 $.
	\end{theorem}

	\begin{corollary}[Zero-energy equation]
 		\label{zerores}
 		\mbox{}	\\
		Under the same assumptions of Theorem \ref{main}, if $\la=\la_j$, then there exists a distributional solution of the zero-energy equation $\HH \Psi = 0$.
	\end{corollary}
Clearly the interesting case is when $c_j$ is different from zero. It turns out that this is true if and only if for $\la=\la_j$, we are not in the exceptional case of the second case, according to 
terminology of Definition \ref{usb}. In other words, we have a Feshbach resonance, if and only if $\HH$  presents a zero-energy resonance; see \cite{ccft} for details.

\section{Low Energy Expansion of the Resolvent}
Corollary \ref{zerores} suggests the resolvent $(\HH - k^2)^{-1}$ is singular as $k\to 0 $ when $\la=\la_j$. In this section we give a characterization of the zero energy eigenspace of $\HH$
and we study some properties of the low energy singularities of the resolvent of $\HH$. In what follows we will assume always that $\lambda\neq |E_{j}|$.

First we recall the Schur-Grushin-Feshbach formula: suppose that $X=X_0  \dot{+} X_1$, direct sum of linear spaces, and that we have a linear operator
$L$ on $X$ given by
\beq 
L= \left(
\begin{array}{cc}
L_{00} & L_{01} \\
L_{10} & L_{11}
\end{array}
\right)
\eeq
with $L_{11}$ invertible. Define $C= L_{11} - L_{10}\, L_{00}^{-1} L_{01}$.
Then $L^{-1} $ exists iff $C^{-1}$ exists and 
\beq \label{sgf}
L^{-1}= \left(
\begin{array}{cc}
C^{-1}  &- C^{-1} L_{10}L_{00}^{-1}  \\
  - L_{00}^{-1} L_{01}C^{-1} & L_{00}^{-1}  + L_{00}^{-1} L_{01}C^{-1} L_{10}L_{00}^{-1} 
\end{array}
\right)
\eeq
We can use formula \eqref{sgf} to give a representation of the resolvent of $\HH$.
In this case we put
\be
C(k)= -\Delta + V - k^2 -W \, R_U (k^2-\la) \, W.
\ee 
It is straightforward to see that $C^{-1}(k)\in {\mathcal B}(L^2)$ exists for $\Im k^2 \neq 0$. Let $k^2=\al + i \beta$ then
\begin{multline*}
\| ( -\Delta + V - k^2 -W \, R_U (k^2-\la) \, W)f \|^2 =\\ \|  ( -\Delta + V -\al  -W \, R_U (k^2-\la) \, W)f \|^2 +\beta^2 \|f\| +2\beta^2 \|R_U (k^2-\la) \, W f \|^2.
\end{multline*}
and
\begin{align}\label{res}
&(\HH-k^2)^{-1} = \nonumber \\
&\left(
\begin{array}{cc}
C^{-1}(k) & -C^{-1}(k)\, W \, R_U (k^2-\la) \\
-R_U (k^2-\la) \, W \, C^{-1}(k) &  R_U (k^2-\la) +  R_U (k^2-\la) \, W \, C^{-1}(k) \, W \, R_U (k^2-\la)
\end{array}
\right).
\end{align}
We can write for $\Im k^2 \neq 0$
\begin{align*}
C(k)& =( -\Delta + V - k^2) \left(I  -R_V (k^2)\, W \, R_U (k^2-\la) \, W \right)\\ &=  (I  - W \, R_U (k^2-\la) \, W \, R_V (k^2)) ( -\Delta + V - k^2)
\end{align*}
then
\begin{align*}
C^{-1}(k) &= (I  -R_V (k^2)\, W \, R_U (k^2-\la) \, W )^{-1} R_V (k^2)\\ &=  R_V (k^2)\, (I  - W \, R_U (k^2-\la) \, W \, R_V (k^2))^{-1}.
\end{align*}
For sake of notation we define
\begin{align*}
M(k) & = R_V (k^2)\, W \, R_U (k^2-\la) \, W \\
N(k) & =   W \, R_U (k^2-\la) \, W \, R_V (k^2)
\end{align*}
so that
\begin{align*}
C^{-1}(k) &= (I  -M(k) )^{-1} R_V (k^2)\\ &=  R_V (k^2)\, (I  - N(k))^{-1}.
\end{align*}
Notice that, $\xx^s W \in L^2$ for some $s>3/2$.
Therefore for $0\leq k^2 <\lambda$, 
\beq \label{strangepot}
 W \, R_U (k^2-\la) \, W \in \B_0 ( H^2_{-s}, L^2_{s} ) ,
\eeq
then we have
\beq
M(k) \in \B_0 ( H^2_{-s}, H^2_{-s} ) \qquad N(k) \in \B_0 ( L^2_{s}, L^2_{s} ) \qquad 
1/2 < s  < 3/2.
\eeq
In \cite{ccft}, it was proved that $k^2 \in \EE$ if and only there exists $\tilde u \in \BB$ such that $ (I  -M(k) ) \tilde u = 0 $.
Notice that $ M(k)  \tilde u \in H^2_{-s}$ for some $s>1/2$ and therefore $  \tilde u \in H^2_{-s}$.
Indeed if $\tilde u \in \BB$ then $W \tilde u \in L^2$, $ R_U (k^2-\la) \, W \, \tilde u \in H^2$, $W\, R_U (k^2-\la) \, W \, \tilde u \in L^2_{s}$ for some some $s>1/2$
and finally $R_V (k^2) W\, R_U (k^2-\la) \, W \, \tilde u \in H^2_{-s}$ for some some $s>1/2$.

Then by Fredholm's alternative $(I  -M(k) )^{-1} \in \B (H^2_{-s}, H^2_{-s} ) $  with $s>1/2$ for $0<k^2 <\la$, $k^2 \notin \EE$, $k^2-\la \neq |E_k|$.
Therefore we have for such values of $k^2$ that  $C^{-1}(k) \in \B (L^2_s, H^2_{-s'} ) $ for $k^2> 0$  if $s,\, s' >1/2$  and the boundary value
of the resolvent is well defined as an operator between suitable weighted spaces by \eqref{res}.


We want to discuss the limit of $k^2\to 0$ of \eqref{res}.
We have that  the existence of $C^{-1}(0)$ is related to the existence of $(I- M(0))^{-1}$ or $(I- N(0))^{-1}$.
For this reason we define 
\begin{align}
\MM&=\lf\{ u\in H^2_{-s} \text{ s.t. } (I-M(0))u=0 \ri\} \\
\NN&=\lf\{ u\in L^2_{s} \text{ s.t. } (I-N(0))u=0 \ri\} 
\end{align}
Both $\MM$ and $\NN$ are finite dimensional for $1/2 < s  < 3/2$ due to the compactness properties pointed out.
\begin{proposition}
The sets $\MM$ and $\NN$ are isomorphic as vector spaces and do not depend on $s$. The linear isomorphisms are given by the restrictions of $ W \, R_U (-\la) \, W$ to $ \MM$ and
$R_V(0)$ to $ \NN$ respectively, that is:
\beq \label{isom}
 W \, R_U (-\la) \, W: \MM \to \NN  \qquad 
R_V(0) : \NN \to \MM .
\eeq
The operator $ W \, R_U (-\la) \, W$  can be substituted by $ -\Delta + V $.
\end{proposition}
\begin{proof}
Let $u\in \MM$ then  $ W \, R_U (-\la) \, W \, u \in \NN$ , indeed
\be 
0=  W \, R_U (-\la) \, W (I-M(0))\, u =  (I-N(0)) \,  W \, R_U (-\la) \, W \, u.
\ee
The map is injective: if $ W \, R_U (-\la) \, W \, u=0$ then we have $u =R_V (0)\, W \, R_U (-\la) \, W \, u=0$. 
Moreover it is clear that $R_V (0)$ is the the inverse of  $ W \, R_U (-\la) \, W$ on the image of $\MM$.
We prove that $ W \, R_U (-\la) \, W $ is also onto. First notice that  $R_V (0)$ on $\NN$ is injective since $(-\Delta + V) R_V (0)= I$ on $\NN$. Then if 
$v\in \NN$ we have $v= W \, R_U (-\la) \, W \,  R_V (0)\, v$  and it is sufficient to prove that $ R_V (0)\, v\in \MM$.
This is straightforward since 
\begin{multline*}
( I -  R_V (0)\,W \, R_U (-\la) \, W ) \,  R_V (0)\, v =\\ 
=  R_V (0)\, ( I -  \,W \, R_U (-\la) \, W  \,  R_V (0) )\, v =0.
\end{multline*}
This proves that $\MM$ and $\NN$ are isomorphic and \eqref{isom}. Since the two spaces have opposite monotony in $s$, they are in facts independent. Notice that
$ W \, R_U (-\la) \, W$ and $ -\Delta + V$ coincides  that on $\MM$ by the definition of $\MM$.
\end{proof}
\begin{proposition} \label{kernel}
We have $C(0) \, \MM =0$ and $ \text{Ker }C(0)=\MM$ in $H^2_{-s}$ for $1/2 < s  < 3/2$.
\end{proposition}
\begin{proof}
Since $\MM \subset H^2_{-s}$ then for $u\in \MM $ we have 
\be
C(0) u = (-\Delta + V) \, (I-M(0) )\, u = 0
\ee
Suppose that $u\in H^2_{-s}$ and $C(0) u =0$. Then $(-\Delta + V) u =  W \, R_U (-\la) \, W \, u \in L^2_s$ by \eqref{strangepot} and $u\in \DD (H_V)$. Hence 
$u= R_V(0) \, (-\Delta + V) u =  R_V(0) \, W \, R_U (-\la) \, W \, u$.
\end{proof}
\begin{proposition}
For $1/2 < s  < 3/2$ there exists operators  $Q,\, K \in \B(H_{-s}^2)$  such that
\begin{align}
& Q^2=Q, \quad QK=KQ=0 \label{propr1} \\
& Q(I-M(0))=(I-M(0))Q=0 \label{propr1.5} \\
& K(I-M(0))=(I-M(0))K=I-Q \label{propr2} \\
& Q \text{ is of finite rank and } K-I \in \BB_0(H_{-s}^2) \label{propr3} \\
&  W \, R_U (-\la) \, W \, K = K\,  W \, R_U (-\la) \, W \quad R_V(0) \, K = K \, R_V(0) \label{propr4}
\end{align}
\end{proposition}
\begin{proof}
Define $Q$ as a spectral projection by the analytic functional calculus,
\beq
Q= -\frac{1}{2\pi i } \int_{|z-1|=\de}  \dfrac{dz}{(M(0)-z)},
\eeq
with $\de$ sufficiently small that $\{|z-1|=\de\}$ only includes the eigenvalue.
This is possible due to the compactness of $M(0)$.
Then $(I-M(0) +Q)$ is invertible and we can define 
\beq
K= (I-M(0) +Q)^{-1} (I-Q).
\eeq
Properties \eqref{propr1}, \eqref{propr1.5} and \eqref{propr2} follow from the separation of spectrum, see \cite{K} pg. 178.
The operator $Q$ is finite rank since $M(0)$ si compact. Taking into account the identity $ (I+B)^{-1} = I - B(I+B)^{-1}$ , we have that
$K-I$ has the same regularity of $M(0)$ and therefore it is compact. The last property \eqref{propr4} can be proved as in Lemma 3.5 in \cite{KJ}.
\end{proof}
Due to Proposition \ref{kernel} and the positivity of $-\Delta +V$ by Assumption \ref{boundst},  we have that $( W \, R_U (-\la) \, W u,  u)$ defines a inner product on $\MM$.
At the same time  we have that $( R_V(0) v , v)$ defines an inner product on $\NN$. Notice that $\MM$ and $\NN$ possess a natural duality induced by $L^2$ inner product
since $H^2_{-s} \subset L^2_{-s} = (L^2_s)^*$. Under this coupling, if $\{ u_j\}$, $j=1,\ldots , d$,  is an orthonormal basis in $\MM$ then $\{v_j\}$ with 
$v_j= W \, R_U (-\la) \, W u_j $ is  orthonormal in $\NN$  and it is also the dual basis with
respect to the above defined inner products. The two bases are orthonormal w.r.t. the two above defined inner products.
Moreover the spectral projector $Q$ reads
\beq
Q= \sum_{j=1}^d  |u_j\ran \lan v_j| .
\eeq
Let us call $\MM_{\EE}= \{ u \in \MM \text{ s.t. } u \in H^2\}$. It is straightforward to check that if $u\in \MM_{\EE}$ then $\tilde \Psi= (u, -R_U (-\la) \, W \, u)^t $ is a zero-energy eigenvalue of $\HH$
that is $\HH \tilde \Psi=0$ and $\tilde \Psi \in \HH$.
We can have different cases
\begin{definition} \label{usb}
We say that we are in the generic case if $\MM = \{0\}$ and in the exceptional case otherwise. In the exceptional case we distinguish the following
situations: in the first kind $\MM\neq \{0 \} $ and $\MM_{\EE} = \{ 0 \}$, in the second case $\MM=\MM_{\EE}\neq \{0 \} $ and in the third case ${0}\subsetneq \MM_{\EE} \subsetneq \MM$.
\end{definition}
\noindent
Notice, see \cite{ccft}, that 
\beq
	\lf(R_V(k^2) w\ri)(\xv) \underset{x \to + \infty}{=}  \braket{\phi_{V,\kvp}}{w} \frac{e^{i k x}}{4\pi x} + O(x^{-2}).
\eeq
Then $u\in \MM$ belongs to $\MM_{\EE}$ iff
\be
 \braket{\phi_{V,0}}{ W \, R_U (-\la) \, W u} = 0
\ee
and in particular $\text{dim } \MM / \MM_{\EE} \leq 1$ that is, there is at most one resonance.
In the main theorem we discuss the expansion of $C^{-1}(k)$.
\begin{theorem}[Low Energy Expansion] \label{ourmain}
\mbox{}	\\
\noindent
Let Assumption \ref{pot} and  Assumption \ref{boundst} hold and let $\{ \la_j \}$ be the critical values as in Theorem \ref{main}.
Then if $\la \notin \{ \la_j \}$ then we are in the generic case and $C^{-1}(0)$ exists. 
Otherwise  if $\la \in \{ \la_j \}$ we are in exceptional case and we have the following expansions :
\begin{itemize}
\item In the exceptional case of first kind, we have:
\beq \label{first}
C^{-1} (k) =  \frac{1}{ak} |u \ran \lan u| + O(1) 
\eeq
in $ \B ( L^2_{s}, H^2_{-s'} ) $ with $s,\, s' >1/2$ and $s+s'>2$
\noindent where 
$$a=  -( W \, R_U (-\la) \, W u ,  R_V ' (0)\,W \, R_U (-\la) \, W u )\neq0.$$
\item In the exceptional case of second kind assume $W\in (I)_5$, then we have:
\beq \label{second}
C^{-1} (k) =   \frac{1}{k^2} P_0     -  \frac{1}{k} P_0  W \, R_U (-\la) \, W  T_3 P_0    +O(1) .
\eeq
in $ \B ( L^2_{s}, H^2_{-s'} ) $ with $s>1/2$, $s'>7/2$,
where $P_0$ is the orthogonal projector on the zero-eigenspace of $\HH$ and $T_3$ is defined by \eqref{term3}. 
\end{itemize}
\end{theorem}
\noindent
We omit the discussion of the exceptional case of the third kind for sake of brevity.
\begin{proof}
The first part of the theorem is just a rephrasing of Theorem \ref{main}:  the critical values $\{\la_j \}$ are exactly the values of $\la$ such $\MM \neq \{0\}$.
Therefore $(I-M(0))^{-1}$ exists by Fredholm's alternative and $C^{-1}(0) = (I-M(0))^{-1} \, R_V(0) \in  \B_0 ( L^2_{s}, H^2_{-s'} ) $ with $s,\, s' >1/2$ and $s+s'>2$.

\noindent
Now we discuss the exceptional case of the first kind.
Let us define $\overline Q= I-Q$ and decompose $M(k)$ accordingly, that is, we put
\beq
M(k)= \lf(
\begin{array}{cc}
\Qb (I-M(k)) \Qb & \Qb(I-M(k)) Q \\
Q (I-M(k)) \Qb & Q(I-M(k)) Q 
\end{array}
\ri) =
\lf(
\begin{array}{cc}
m_{00}(k) & m_{01}(k)\\
m_{10}(k) & m_{11}(k)
\end{array}
\ri)
\eeq
Remember that in this case $Q= | u \ran \lan W \, R_U (-\la) \, W u |$ with $u\in \MM$; we normalize $u$ requiring that $( u,  W \, R_U (-\la) \, W u )=1$.
Notice that $m_{00} (k)$ is  continuous and $m_{00}^{-1}(0)$ exist by construction. Then $m_{00}^{-1}(k)$ exists and it is continuous for $k$ sufficiently small
by Neumann series. By \eqref{sgf} we have to discuss the invertibility of 
\beq 
m(k)   = m_{11}(k) - m_{10 }(k)  m_{00}^{-1}(k)  m_{01}(k) 
\eeq 
on the range of $Q$.
In facts in this case we have
\begin{multline*}
m(k) = Q \lf[ ( W \, R_U (-\la) \, W u , (I-M(k)) u )+ \ri.\\ -\lf.  ( W \, R_U (-\la) \, W u ,   (I-M(k))  \Qb  m_{00}^{-1}(k)  \Qb (I-M(k)) u ) \ri]
\end{multline*}
The resolvent $R_V (k^2)$ has the expansion in $\B ( L^2_{s}, H^2_{-s'} ) $ with $s,\, s' >1/2$ and $s+s'>2$,  see Lemma 2.3 in\cite{KJ} and Lemma 2.3 in \cite{Y}.
\[
R_V (k^{2}) = R_V (0) + R_v ' (0) \, k + o(k)
\]
where $'$ denote the derivative w.r.t. $k$. Notice also the analiticity of $R_U( k^2-\la)$ in $k^2$ in   $\B ( L^2, H^2 ) $.
Then the following expansion  in  $\B_0 ( H^2_{-s}, H^2_{-s} ) $, $1/2<s<3/2$, holds true: 
\beq \label{exp}
M(k) =  R_V (0)\,W \, R_U (-\la) \, W + k  R_V ' (0)\,W \, R_U (-\la) \, W +   o(k).
\eeq
Then we have
\begin{align*}
& ( W \, R_U (-\la) \, W u , (I-M(k)) u ) = 1 - ( W \, R_U (-\la) \, W u , M(0) u ) \\
 & - k ( W \, R_U (-\la) \, W u ,  R_V ' (0)\,W \, R_U (-\la) \, W u ) +o(k) \\
&= a\, k + o(k),
\end{align*}
where we have put
\[
a=  -( W \, R_U (-\la) \, W u ,  R_V ' (0)\,W \, R_U (-\la) \, W u )
\]
Since 
\beq \label{deri}
R_V ' (k^2) = (I- R_V (k^2) V) R_0 ' (k^2)  (I-V R_V (k^2) )
\eeq
then we have
\begin{align*}
a&=\frac{1}{4\pi}| (1,  (I-V R_V (0) ) W \, R_U (-\la) \, W u ) |^2 \\
  &= \frac{1}{4 \pi}| (\phi_{V,0} ,W \, R_U (-\la) \, W u ) |^2 \neq 0
\end{align*}
The constant $a$ is different from 0 otherwise $u$ would be a 0-energy eigenstate and we would be in the exceptional case of the second kind. 
Using again expansion \eqref{exp}, we have also $(I-M(k)) u= O(k)$
and  $ ( W \, R_U (-\la) \, W u ,   (I-M(k)) f) =  ((I-N(k)) W \, R_U (-\la) \, W u ,    f) = O(k)$. Therefore we obtain
\[
m(k)= a k Q+  o(k) \qquad   m^{-1} (k) = \frac{1}{ak} Q +o(1)
\]
Using \eqref{sgf} and the above remarks, we see that the only singular terms comes from the term $d^{-1}$ and we obtain 
\[
(I-M(k) )^{-1} = \frac{1}{ak} Q +O(1) \quad \text{ in }  \B_0 ( H^2_{-s}, H^2_{-s} ) \quad 1/2<s<3/2
\] 
Then 
\beq
C^{-1} (k) =  \frac{1}{ak} |u \ran \lan u| + O(1) \quad \text{ in }  \B_0 ( L^2_{s},  H^2_{-s'} ) \quad s,\, s' >1/2 \quad s+s'>2
\eeq
and this proves \eqref{first}.

\noindent
Now we consider the exceptional case of the second case. Again the main point is the inversion of $m(k)$ on the range of $Q$.
In this case  we have
\beq \label{canc}
\begin{split}
&Q= \sum_{j=1}^d |u_j\ran \lan W \, R_U (-\la) \, W \, u_j| \\
&(1, (I-V R_V(0)) \, W \, R_U (-\la) \, W \, u_j) = 0 \;\;\;j=0,\ldots,d
\end{split}
\eeq
Let us start expanding around $k=0$
\beq \label{m11}
m_{11}(k)= Q (I-M(k))Q= \sum_{j,k=1}^d |u_j\ran \lan v_k| (v_j, (I-M(k)) u_k)
\eeq
Taking into account $W\in (I)_5$, in the following of the proof we can choose $s>7/2$ such that $\xx^s W \in L^2$.
Then the following expansion holds true in $\B_0 ( H^2_{-s}, H^2_{-s} ) $ with $s>7/2$:
\begin{align}
I-M(k) & = I-M(0)+ \label{term0} \\
&-k R_V ' (0)  W \, R_U (-\la) \, W \label{term1} \\
&-\frac{k^2}{2} \lf( R_V '' (0)  W \, R_U (-\la) \, W  + 2 R_V  (0)  W \, R_U^2 (-\la) \, W \ri) \label{term2} \\
&-\frac{k^3}{6} \lf( R_V ''' (0)  W \, R_U (-\la) \, W  + 6 R_V'  (0)  W \, R_U^2 (-\la) \, W \ri) \label{term3} \\
& +O(k^4) \nonumber
\end{align}
Both \eqref{term0} and \eqref{term1} do not contribute in the expansion of \eqref{m11}: the former since $u_j\in \MM$ and the latter
due to \eqref{deri} and the cancellation condition in \eqref{canc} (remember that $R_0'(0) =1$). Now we discuss the 
first term in \eqref{term3}: for the same reasons as in the analysis of \eqref{term1} we have
\begin{align}
& (W \, R_U (-\la) \, W u_j, R_V '' (0)  W \, R_U (-\la) \, W  u_k) =\nonumber \\
&  (W \, R_U (-\la) \, W u_j,  (I- R_V (0) V) R_0 '' (0)  (I-V R_V (0) )  W \, R_U (-\la) \, W  u_k) =\nonumber  \\
&  2( R_0  (0)  (I-V R_V (0) )  W \, R_U (-\la) \, W u_j,   R_0  (0)  (I-V R_V (0) )  W \, R_U (-\la) \, W  u_k) =\label{special} \\
&  2 ( R_V  (0)   W \, R_U (-\la) \, W u_j,   R_V  (0)   W \, R_U (-\la) \, W  u_k)= \nonumber \\
& 2(u_j, u_k) 
\end{align}
Notice that in \eqref{special} we have used Lemma 2.6 of \cite{KJ} and the cancellation condition in \eqref{canc}.
For the second term \eqref{term3} we have 
\[
  (W \, R_U (-\la) \, W u_j,  R_V  (0)  W \, R_U^2 (-\la) \, W  u_k) =  ( u_j,   W \, R_U^2 (-\la) \, W  u_k)
\]
Define the matrix $A$ by
\[
A_{j,k} = (u_j, u_k) + ( u_j,   W \, R_U^2 (-\la) \, W  u_k).
\]
Since $A$ is positive definite, we can define $B=A^{-1/2}$ and $\tilde u_k = \sum_j B_{k,j} u_j$ .
Then $\{\tilde u_j \}$ is orthonormal basis of $\MM$ w.r.t. the $L^2$ inner product and 
\[
\lf(  \sum_{j,k=1}^d |u_j\ran \lan v_k| A_{j,k}  \ri)^{-1}
= 
\sum_{j=1}^d |\tilde u_j\ran \lan \tilde u_j|   W \, R_U (-\la) \, W = P_0 \,  W \, R_U (-\la) \, W
\]
By the same arguments we have
\[
m_{10 }(k)  m_{00}^{-1}(k)  m_{01}(k) = O(k^4)
\]
Collecting the above results we have
\begin{align*}
& m^{-1} (k)  = \frac{1}{k^2} P_0  W \, R_U (-\la) \, W \lf( I + k T_3 P_0  W \, R_U (-\la) \, W  +O(k^2)\ri)^{-1} =\\
&  \frac{1}{k^2} P_0  W \, R_U (-\la) \, W    -  \frac{1}{k} P_0  W \, R_U (-\la) \, W k T_3 P_0  W \, R_U (-\la) \, W  +O(1)
\end{align*}
where we have denoted by $T_3$ the expression in \eqref{term3}. Using \eqref{sgf} and the above remarks, we see that the only singular terms comes from the term $d^{-1}$ and we obtain 
\begin{align*}
&(I-M (k))^{-1}  = \\
&\frac{1}{k^2} P_0  W \, R_U (-\la) \, W    -  \frac{1}{k} P_0  W \, R_U (-\la) \, W k T_3 P_0  W \, R_U (-\la) \, W  +O(1)
\end{align*}
in $ \B_0 ( H^2_{-s}, H^2_{-s} ) $ $s>7/2$.
Since $P_0  W \, R_U (-\la) \, W \, R_V(0)= P_0$ wei finally obtain
\[
C^{-1} (k)  = \frac{1}{k^2} P_0     -  \frac{1}{k} P_0  W \, R_U (-\la) \, W T_3 P_0    +O(1)
\]
in $ \B_0 ( L^2_{s}, H^2_{-s'} ) $ with $s>1/2$, $s'>7/2$.
which proves \eqref{second}
\end{proof}
We expect that the hypothesis on the potential are not optimal w.r.t. to the decay at infinity. We have discussed the most simple expansion of the resolvent while
leaving untouched the differentiability of the remainder which is a central issue in the proof of dispersive estimates and boundedness properties of the wave
operators.

\end{document}